\documentclass[11pt,index]{amsart}

\usepackage{amssymb,amsmath}
\usepackage{mathtools}

\usepackage{enumerate,xcolor,mathrsfs}
\usepackage[pagebackref,colorlinks,linkcolor=red,citecolor=blue,urlcolor=blue,hypertexnames=true]{hyperref}

\usepackage{ulem,xspace} 

\usepackage{comment}

\includecomment{fornow}\excludecomment{maybelater}

\renewcommand{\subset}{\subseteq}
\makeatletter
\newsavebox\myboxA
\newsavebox\myboxB
\newlength\mylenA
\usepackage{graphicx}

\newcommand*\xoverline[2][0.75]{%
    \sbox{\myboxA}{$\m@th#2$}%
    \setbox\myboxB\null
    \ht\myboxB=\ht\myboxA%
    \dp\myboxB=\dp\myboxA%
    \wd\myboxB=#1\wd\myboxA
    \sbox\myboxB{$\m@th\overline{\copy\myboxB}$}
    \setlength\mylenA{\the\wd\myboxA}
    \addtolength\mylenA{-\the\wd\myboxB}%
    \ifdim\wd\myboxB<\wd\myboxA%
       \rlap{\hskip 0.5\mylenA\usebox\myboxB}{\usebox\myboxA}%
    \else
        \hskip -0.5\mylenA\rlap{\usebox\myboxA}{\hskip 0.5\mylenA\usebox\myboxB}%
    \fi}
\makeatother

\newtheorem{theorem}            {Theorem}[section]
\newtheorem{corollary}          [theorem]{Corollary}
\newtheorem{proposition}        [theorem]{Proposition}

\newtheorem{remark}         [theorem]{Remark}

\newcommand{\cR}{\mathcal{R}}

\newcommand{\cM}{\mathcal{M}}
\newcommand{\cN}{\mathcal{N}}
\newcommand{\cH}{\mathcal{H}}

\makeatletter
\def\moverlay{\mathpalette\mov@rlay}
\def\mov@rlay#1#2{\leavevmode\vtop{
		\baselineskip\z@skip \lineskiplimit-\maxdimen
		\ialign{\hfil$#1##$\hfil\cr#2\crcr}}}
\makeatother

\newcommand{\CC}{\mathbb{C}}
\newcommand{\RR}{\mathbb{R}}

\newcommand{\cI}{\mathcal{I}}

\newcommand{\cP}{\mathcal{P}}
\newcommand{\trace}{\operatorname{tr}}
\usepackage{color}

\usepackage{ulem}
\usepackage{upgreek}

\newcommand{\vertiii}[1]{{\left\vert\kern-0.25ex\left\vert\kern-0.25ex\left\vert #1 
    \right\vert\kern-0.25ex\right\vert\kern-0.25ex\right\vert}}

\newcommand{\wXi}{\widetilde{\Xi}}

\makeindex

 \title[Convexity]{Convexity of a certain operator trace functional}

  \author[Evert]{ Eric Evert${}^1$}
 \address{Group Science, Engineering and Technology\\
 KU Leuven Kulak \\
  E. Sabbelaan 53, 8500 Kortrijk, Belgium \\
  and
  \newline
   Electrical Engineering ESAT/STADIUS\\
  KU Leuven \\
  Kasteelpark Arenberg 10, 3001 Leuven, Belgium
  }
 \email{eric.evert@kuleuven.be}
\thanks{${}^1$Supported by the Fonds de la Recherche Scientifique--FNRS and the Fonds Wetenschappelijk Onderzoek--Vlaanderen under EOS Project no 30468160 (SeLMA)}

 \author[McCullough]{Scott McCullough}
 \address{Department of Mathematics\\
   University of Florida \\ Gainesville, FL}
 \email{sam@ufl.edu}

  \author[Strekelj]{Tea  \v{S}trekelj}
  \address{Institute of Mathematics, Physics and Mechanics, Ljubljana, Slovenia}
  \email{tea.strekelj@fmf.uni-lj.si}

 \author[Vershynina]{Anna Vershynina${}^2$}
 \address{Department of Mathematics\\
 University of Houston \\ TX}
 \email{anna@math.uh.edu}
 \thanks{${}^2$Supported by NSF grants DMS-1812734, DMS-2105583.  }


\begin{document}
\maketitle
 \begin{abstract}
In this article the operator trace function $ \Lambda_{r,s}(A)[K, M] := \trace(K^*A^r M A^r K)^s$ is introduced and its convexity and concavity properties are investigated.
This function has a direct connection to several well-studied operator trace functions that  appear in quantum information theory, in particular when studying data processing inequalities of various relative entropies.  In the paper the interplay between $\Lambda_ {r,s}$ and the well-known operator functions $\Gamma_{p,s}$ and $\Psi_{p,q,s}$ is used to study the stability of their convexity (concavity) properties. This interplay may be used to ensure that $\Lambda_{r,s}$ is convex (concave) in certain parameter ranges when $M=I$ or $K=I.$ However, our main result shows that convexity (concavity) is  surprisingly lost when perturbing those matrices even a little.  To complement the main theorem, the convexity (concavity) domain of $\Lambda$ itself is examined. The final result states that $\Lambda_{r,s}$ is never concave and it is convex if and only if $r=1$ and $s\geq 1/2.$ 
 \end{abstract}

\section{Introduction}

We consider convexity of several operator trace functions motivated by problems in quantum information theory. In particular, we consider the data processing inequality (or monotonicity) for several relative entropies. The inequality effectively states that quantum states become harder to distinguish after they pass through a noisy quantum channel \cite{BDK08, HP91}. It has been shown that this inequality for some relative entropies is equivalent to convexity of certain trace functions. Additionally, the convexity of these trace functions  give rise to the conditions on states that ensure  equality in the data processing inequality.

\subsection{Umegaki relative entropy.} Let $\cH$ be a finite-dimensional Hilbert space and let  $\rho$ and $\sigma$ be two density matrices (states) on $\cH.$ The Umegaki relative entropy is defined as
\[
 S(\rho\|\sigma):=\trace(\rho\log\rho-\rho\log\sigma)\ , 
\]
when the null space of $\sigma$ is contained in null space of $\rho,$ and $+\infty$ otherwise. This quantity was defined in 1962 \cite{U54} as a direct generalization of a classical relative entropy known as the Kullback-Leibler divergence \cite{B97}.
 Umegaki relative entropy is sometimes referred to as quantum relative entropy or just relative entropy.

In 1975 Lindblad \cite{L75}, building on the work of Lieb and Riskai \cite{LR73}, proved that the Umegaki relative entropy satisfies the data processing inequality for quantum channels. Let $\cN$ be a completely positive trace preserving map, that is, a quantum channel, mapping states between two Hilbert spaces. Then the following holds
\begin{equation}\label{eq:DPI}
S(\cN(\rho)\|\cN(\sigma))\leq S(\rho,\sigma)\ .
\end{equation}
This inequality is commonly referred by two names: data processing inequality (DPI) and monotonicity inequality.

In 1986-88, Petz \cite{P86, P88} classified all states and quantum channels that saturate the data processing inequality.  Here we say the states $\rho$ and $\sigma$ saturate the data processing inequality
if equality holds in the DPI \eqref{eq:DPI} for these states.  Petz showed that for a given quantum channel $\cN,$ two quantum states $\rho$ and $\sigma$ saturate the DPI (\ref{eq:DPI}) if and only if both states can be recovered by a map $\cR,$ known as a Petz recovery map. The map $\cR$  has the following explicit form
\begin{equation}
\label{e:Petz}
 \cR_{\rho, \cN}(\omega)=\rho^{1/2}\cN^*\left(\cN(\rho)^{-1/2}\omega\, \cN(\rho)^{-1/2} \right)\rho^{1/2}\ . 
\end{equation}
It is trivial that  $\cR_{\rho,\cN}$ always recovers $\rho$ perfectly; that is,  $\cR_{\rho,\cN}(\cN(\rho))=\rho.$ However, the recovery of $\sigma$ via $\cR_{\rho, \cN}(\sigma)=\sigma$ cannot be generally expected.  In other words, Petz showed that 
\[
 \cR_{\rho, \cN}(\sigma)=\sigma\ \Longleftrightarrow  \ S(\cN(\rho)\|\cN(\sigma))= S(\rho,\sigma)\ .
\]
One additional point to note is that the recovery condition is symmetric in $\rho$ and $\sigma;$ that is, $\cR_{\rho, \cN}(\sigma)=\sigma$ if and only if $\cR_{\sigma, \cN}(\rho)=\rho.$

\subsection{R\'enyi relative entropy.} 
{ For $\alpha\in(-\infty, 1)\cup(1, +\infty),$ 
the quantum R\'enyi relative entropy (RRE), denoted $S_\alpha,$
 is a generalization of a classical R\'enyi divergence. For 
 states $\rho$ and  $\sigma,$ it is defined by
 \[
 S_\alpha(\rho\|\sigma):=\frac{1}{\alpha-1}\log\trace\left( \rho^\alpha\sigma^{1-\alpha}\right)\ 
\]
 when the null space of $\sigma$ is contained in null space of $\rho,$ and $+\infty$ otherwise. In the limit $\alpha\rightarrow1,$ the R\'enyi relative entropy approaches the Umegaki relative entropy.
 
R\'enyi entropy relative satisfies the data processing inequality
\[
 S_\alpha(\cN(\rho)\|\cN(\sigma))\leq S_\alpha(\rho,\sigma) 
\]
for  $\alpha\in[0,1)\cup(1,2].$ This inequality was  was established for  a larger class of quantum relative entropies, called quantum $f$-divergences, in \cite{HMH09, P86-1}.  Saturation of the data processing inequality for the R\'enyi relative entropy is again equivalent to the ability to recover both states after they pass through a quantum channel, and the recovery map is the same one that appears in equation~\eqref{e:Petz}}\cite{HMPB11}. Some other equivalent conditions on  pairs of states that saturate the DPI for R\'enyi relative entropy were given by Hiai et al in the same reference \cite{HMPB11}.
 
\subsection{Sandwiched R\'enyi relative entropy.} Another way to generalize the classical R\'enyi divergence is to take into account the noncommutativity of $\rho$ and $\sigma.$ In 2013 the following quantity was proposed independently by Wilde et al \cite{WWY14} and M\"uller-Lennert et al \cite{MDSFT13}: for $\alpha\in(0, 1)\cup(1, +\infty),$ the sandwiched R\'enyi relative entropy (sandwiched RRE) is defined as
\[
\tilde{S}_\alpha(\rho\|\sigma):=\frac{1}{\alpha-1}\log\left( \trace\left[ \sigma^{\frac{1-\alpha}{2\alpha}}\rho \sigma^{\frac{1-\alpha}{2\alpha}} \right]^\alpha \right)\ ,
\]
 when the null space of $\sigma$ is contained in null space of $\rho,$ and $+\infty$ otherwise.  In the limit $\alpha\rightarrow1,$ the sandwiched R\'enyi relative entropy approaches the Umegaki relative entropy.

 The quantity inside the logarithm does not belong to the class of $f$-divergences, so previous results on the DPI and its saturation cannot be applied to the sandwiched RRE. Several papers in 2013 provided proofs for the data processing inequality for the sandwiched RRE: for $\alpha\in(1,2]$ in the papers Wilde et al \cite{WWY14} and M\"uller-Lennert et al \cite{MDSFT13}; for $\alpha\in(1,\infty)$ in \cite{B13} by Beigi; and independently of Beigi, Frank and Lieb proved the inequality for $\alpha\in [1/2, 1)\cup (1,\infty)$ in \cite{FL13}.  In \cite{FL13} it was shown that the data processing inequality follows from joint convexity or concavity of the following trace map:
\begin{equation}
 \label{e:FLmap}
 (\rho,\sigma)\mapsto \trace\left(\sigma^{(1-\alpha)/(2\alpha)}\rho \sigma^{(1-\alpha)/(2\alpha)} \right)^\alpha\ .
\end{equation}
 Frank and Lieb showed that the map in equation~\eqref{e:FLmap} is jointly concave for $1/2\leq \alpha<1$ and jointly convex for $\alpha>1,$ by using a variational method to separate the arguments and then  showing that  convexity holds in each of $\rho$ and $\sigma$ independently. This proof relied on the following  result. For  a fixed operator $K,$  the map on positive operators
 \begin{equation}
  \label{e:KAK}
 A\mapsto \trace(K^*A^p K)^{1/p}
 \end{equation}
 is concave for $-1\leq p\leq 1,$ $p\neq 0.$ The concavity of  the map in equation~\eqref{e:KAK} was shown by Epstein \cite{E73} for $0<p\leq 1$ (an alternative proof is given in \cite{CL08} by Carlen-Lieb) and by Frank and Lieb \cite{FL13} for $-1\leq p<0.$ 
 
 As a generalization of the  map~\eqref{e:KAK}, given an operator $K$ and $p,s\in\RR,$ define, for positive operators $A,$
 \begin{equation}\label{def:gamma}
 \Gamma_{p,s}(A)[K]:=\trace(K^* A^p K)^s\ .
 \end{equation}
We note that when considering properties of $\Gamma_{r,s}$ that hold for all $K,$ it is sufficient to consider $s>0,$ since
 $\Gamma_{p,s}(A)[K] =\Gamma_{-p,-s}(A)[(K^*)^{-1}]$
 and invertible operators are dense.

 \begin{theorem}\label{thm:gamma} Let $s>0$ and $K$ be any operator.
 \begin{enumerate}[(a)]
 \item If $0\leq p\leq 1$ and $0<s\leq 1/p,$ then $\Gamma_{p,s}(A)[K]$ is concave in $A$ \cite{CL08, E73, H13}.
 \item If $-1\leq p\leq 0$ and $s>0,$ then $\Gamma_{p,s}(A)[K]$ is convex in $A$ \cite{H13}.
 \item If $1\leq p\leq 2$ and $s\geq 1/p,$ then $\Gamma_{p,s}(A)[K]$ is convex in $A$ \cite{CL08}.
 \end{enumerate}
  \end{theorem} 
 These parameter conditions are also necessary in the sense described
 by Theorem~\ref{thm:gamma-necessary} immediately below. 
 For positive integers $n,$ let $\cM_n$ and $\cP_n$ denote the set of $n\times n$ matrices and positive
 definite $n\times n$ matrices with complex entries respectively.
 
 \begin{theorem}[Hiai \cite{H13}]\label{thm:gamma-necessary} 
 Let $s>0$ and $p\neq 0.$
 \begin{enumerate}[(a)]
 \item If, for each invertible $K\in \cM_2,$  the mapping
   $\Gamma_{p,s}(A)[K],$ as a function on $\cP_2,$ is concave, 
      then $0<p\leq 1$ and $0<s\leq 1/p.$
 \item  If, for each invertible $K\in \cM_4,$ the mapping  $\Gamma_{p,s},$ as a function on $\cP_4,$  is convex, 
     then either $-1\leq p<0$ and $s>0$ or $1\leq p\leq 2$ and $s\geq 1/p.$
  \end{enumerate}
 \end{theorem}
 
 After the data processing inequality for sandwiched RRE was proved for $\alpha\geq 1/2,$ the condition for equality was found by Leditzky et al \cite{LRD17} in 2017 for the same parameter range. It was shown that there is an equality in the data processing inequality if and only if
 \begin{equation}
 \sigma^{\frac{1-\alpha}{2\alpha}}(\sigma^{\frac{1-\alpha}{2\alpha}} \rho \sigma^{\frac{1-\alpha}{2\alpha}} )^{\alpha-1}\sigma^{\frac{1-\alpha}{2\alpha}}= \cN^*\left( \cN(\sigma)^{\frac{1-\alpha}{2\alpha}}\left[\cN(\sigma)^{\frac{1-\alpha}{2\alpha}} \cN(\rho) \cN(\sigma)^{\frac{1-\alpha}{2\alpha}} \right]^{\alpha-1}\cN(\sigma)^{\frac{1-\alpha}{2\alpha}} \right)\ .
 \end{equation}

\subsection{$\alpha-z$ R\'enyi relative entropy.} A more general class of R\'enyi  relative entropy that plays into the noncommutativity of quantum states was introduced by Audenaert and Datta \cite{AD15}. For $\alpha\neq 1$ and $z>0,$ the two-parameter family of $\alpha-z$ R\'enyi relative entropy ( $\alpha-z$ RRE) is defined as
\begin{equation}\label{def:alpha-z-RRE}
S_{\alpha,z}(\rho\|\sigma):=\frac{1}{\alpha-1}\log \left[\trace\left(\sigma^{(1-\alpha)/2z} \rho^{\alpha/z}\sigma^{(1-\alpha)/2z}\right)^z \right]\ .
\end{equation}
For commuting $\rho$ and $\sigma,$ the $\alpha-z$ RRE reduces to the classical R\'enyi divergence for all values of $z.$ Moreover, for $z=1,$ the $\alpha-z$ RRE becomes the quantum RRE, namely,
$S_{\alpha, 1}(\rho\|\sigma)=S_\alpha(\rho\|\sigma).$
And for $\alpha=z,$ the $\alpha-z$ RRE reduces to the sandwiched RRE,
$S_{\alpha,\alpha}(\rho\|\sigma)=\tilde{S}_\alpha(\rho\|\sigma). $
For a comprehensive discussion of various particular cases of the $\alpha-z$ RRE see \cite{AD15}.

In \cite{Z20} the precise parameter range for which the $\alpha-z$ RRE satisfies the data processing inequality was determined, resolving a conjecture from \cite{AD15}.  Carlen et al \cite{CFL18} showed that the data processing inequality is equivalent to the joint convexity/concavity property of the trace functional  appearing in the definition~\eqref{def:Psi} below.  Consider the following trace functional defined on $\cP_n\times \cP_n$ for a fixed operator $K,$ by
\begin{equation}\label{def:Psi}
\Psi_{p,q,s}(A,B)[K]:=\trace\left(B^{q/2}K^* A^p K B^{q/2} \right)^s, \ \ p,q,s\in\RR\ .
\end{equation}
Carlen et al \cite{CFL18} showed that  $\alpha-z$ RRE is monotone under quantum channels if and only if $\Psi_{p,q,s}(A,B)[I]$
is jointly concave for $\alpha<1$ and jointly convex for $\alpha>1,$ where $p=\frac{\alpha}{z},$ $q=\frac{1-\alpha}{z},$ and $s=\frac{1}{p+q}.$

The study of the  map $\Psi_{p,q,s}$ started with the seminal Lieb Concavity Theorem \cite{L73}, long before it was connected to the monotonicity of $\alpha-z$ RRE. Lieb's Concavity Theorem  states that for fixed $K$ and $0\leq p,q,\leq 1,$ $p+q\leq 1$ the function $\Psi_{p,q,1}(A,B)[K]$ is jointly convex in $A$ and $B.$ Following this work, a series of results studying the joint convexity of $\Psi_{p,q,s}$ using different methods  was produced \cite{A79, B04, CFL16,  CL08, E73, FL13, H13, H16, Z20}. Finally Zhang \cite{Z20} completed the classification of the parameter range where $\Psi_{p,q,s}$ is jointly convex/concave.  As with the map $\Gamma_{p,q},$ due to symmetry, it suffices to consider $s>0,$ since
\[
 \Psi_{p,q,s}(A,B)[K]=\Psi_{-p,-q,-s}(A,B)[(K^*)^{-1}]\ .
\]
For simplicity one may also consider $q\leq p,$ since
$\Psi_{p,q,s}(A,B)[K]=\Psi_{q,p,s}(B,A)[K^*]. $

\begin{theorem}[Zhang \cite{Z20}]\label{thm:Psi-conv} 
 Given an invertible operator $K,$ the function $\Psi_{p,q,s}(\cdot,\cdot)[K]$ is
\begin{enumerate}[(a)]
\item jointly concave if $0\leq q\leq p\leq 1$ and $0<s\leq \frac{1}{p+q}$;
\item jointly convex if $-1\leq q\leq p\leq 0$ and $s>0$;
\item jointly convex if $-1\leq q\leq 0,$ $1\leq p\leq 2,$ $(p,q)\neq (1,-1)$ and $s\geq \frac{1}{p+q}.$
\end{enumerate}
\end{theorem}

The necessity of these parameter conditions was shown even before Theorem \ref{thm:Psi-conv} was completed.

\begin{theorem}[Hiai \cite{H13}, Carlen, Frank, Lieb \cite{CFL18}] \label{thm:Psi-conv-nec}  
Let $q\leq p$ and $s>0$ be given. Suppose that $(p,q)\neq(0,0)$ and $K=I.$
\begin{enumerate}[(a)]
\item If $\Psi_{p,q,s}$ is jointly concave on $\CC^2,$ then $0\leq q\leq p\leq 1$ and $0<s\leq \frac{1}{p+q}.$
\item \label{i:Psi-conv-nec2} 
   If $\Psi_{p,q,s}$ is jointly convex on $\CC^4,$ then either $-1\leq q\leq p\leq 0$ and $s>0$ or $-1\leq q\leq 0,$ $1\leq p\leq 2,$ $(p,q)\neq (1,-1)$ and $s\geq \frac{1}{p+q}$\ .
\end{enumerate}
\end{theorem}

Theorem \ref{thm:Psi-conv} together with the equivalence proved in \cite{CFL18} immediately implies that $\alpha-z$ RRE is monotone under quantum channels if and only if one of the following holds:

\begin{enumerate}[(a)]
\item $0<\alpha<1$ and $z\geq \max\{\alpha,1-\alpha\}$;
\item $1<\alpha\leq 2$ and $\alpha/2\leq z\leq \alpha$;
\item $2\leq \alpha<\infty$ and $\alpha-1\leq z\leq \alpha.$
\end{enumerate}

Unlike with Umegaki, R\'enyi and sandwiched relative entropies, there is 
 no known algebraic condition certifying equality in the data processing inequality for the $\alpha-z$ RRE.
 However, different necessary and sufficient conditions have been found.

Chehade \cite{Che20} showed that for $1<\alpha\leq 2$ and $\alpha/2\leq z\leq \alpha,$ if there is  equality in the data processing inequality for $\alpha-z$ RRE and  the following trace function  defined on positive operators
\[
\Uplambda_{p}(A)[K, M] :=\trace\left[\left\{K^* A^{p/2} M A^{p/2} K\right\}^{1/p} \right]
\]
is concave for $p=\frac{\alpha-z}{\alpha} \in (0,1)$ and any positive definite $M$ and any invertible $K,$ then the states satisfy the following condition
\[
  \sigma^{\frac{1-z}{2z}}\left(\sigma^{\frac{1-\alpha}{2z}} \rho^{\frac{\alpha}{z}} \sigma^{\frac{1-\alpha}{2z}} \right)^{z-1}\sigma^{\frac{1-z}{2z}}= \cN^*\left( \cN(\sigma)^{\frac{1-z}{2z}}\left[\cN(\sigma)^{\frac{1-\alpha}{2z}} \cN(\rho)^{\frac{\alpha}{z}}  \cN(\sigma)^{\frac{1-\alpha}{2z}} \right]^{z-1}\cN(\sigma)^{\frac{1-z}{2z}} \right)\ .
\]
Below, as a consequence of Theorem \ref{thm:main}, we show that $\Uplambda_p$ is never concave.
For emphasis, the assumptions used in \cite{Che20} imply $1<z\leq\alpha\leq 2z,$ and hence $0<p<1.$

To finish the picture, necessary and sufficient conditions for equality in the DPI for  $\alpha-z$ RRE were found by Zhang \cite{Z20-2}.  Namely,
\[
\sigma^{\frac{1-\alpha}{2z}}\left(\sigma^{\frac{1-\alpha}{2z}} \rho^{\frac{\alpha}{z}} \sigma^{\frac{1-\alpha}{2z}} \right)^{\alpha-1}\sigma^{\frac{1-\alpha}{2z}}= \cN^*\left( \cN(\sigma)^{\frac{1-\alpha}{2z}}\left[\cN(\sigma)^{\frac{1-\alpha}{2z}} \cN(\rho)^{\frac{\alpha}{z}}  \cN(\sigma)^{\frac{1-\alpha}{2z}} \right]^{\alpha-1}\cN(\sigma)^{\frac{1-\alpha}{2z}} \right)\ . 
\]
Moreover, different conditions were found in \cite{Che20, Z20-2} that imply the equality in the DPI.

\subsection{Operator convexity.} The following trace function was considered by Carlen et al \cite{CFL16}
\begin{equation}\label{def:trace-function}
\Omega_{p,q,r}(A,B,C)=\trace A^{q/2}B^pA^{q/2}C^r .
\end{equation}
These authors showed that if, for fixed  $p,q,r\neq 0,$ the function $\Omega_{p,q,r}$ is jointly convex,  then the following operator function
\begin{equation}\label{def:operator-function}
(A,B)\mapsto  A^{q/2}B^pA^{q/2}
\end{equation}
is operator convex.

It has long been known \cite{LR74} that the function in \eqref{def:operator-function} is jointly convex for $q=2$ and $-1\leq p<0.$ As it turns out \cite{CFL16}, this is the only parameter range where this function is jointly operator convex, and the function is never jointly operator concave.  It follows that $\Omega_{p,q,r}$ is never concave, and it is convex if and only if $q=2,$ $p,r<0$ and $-1\leq p+r<0.$ The later case was proved previously in \cite{L73}.

\subsection{Main results}
To investigate the stability of convexity (concavity) of the trace functions $\Gamma_{p,s}$ and $\Psi_{p,q,s}$ as well as the convexity (concavity) of $\Uplambda_p,$ we introduce the following trace function  defined on positive operators
\[
\Lambda_{r,s}(A)[K, M]: = \trace\left[\left\{K^*A^r M A^r K\right\}^s\right], \  \ r,s\in\RR.
\]
Crucial in establishing our main results is the interplay between $\Lambda_{r,s}$ and the other trace functions introduced so far. Namely, 
\begin{align}
	\Lambda_{r,s}(A)[K,I]&= \Gamma_{2r,s}(A)[K], \label{eq:Lambda-Gamma} \\
	\Lambda_{r,s}(A)[I, M]&=\Psi_{1,2r,s}(M,A)[I], \label{eq:Lambda-Psi} \\
	\Lambda_{r,1}(A)[K,M] & =  \Omega_{1, 2r, 1}(A,M,KK^*),\ 
\nonumber  \\
	\Lambda_{\frac p2,\frac1p}(A)[K,M] & =\Uplambda_p(A)[K,M] .
\nonumber
\end{align}

Our main result is Theorem \ref{th:stab - gamma}, stating that although for fixed K, $\Lambda_{r,s}(\cdot)[K, I]=\Gamma_{2r,s}(\cdot)[K]$ is convex (concave) in some parameter range by Theorem \ref{thm:gamma}, the convexity (concavity) is lost for specific choices of $M$ in any small neighbourhood of the identity.  Corollary \ref{cor:stab - Psi} then gives the analog of Theorem \ref{th:stab - gamma} for the function $\Psi.$ For fixed $M$,
 the  function  $\Lambda_{r,s}(\cdot)[I, M]=\Psi_{1,2r,s}(M,\cdot)[I]$ is convex (concave) in some parameter range by Theorem \ref{thm:Psi-conv}, but in  any small neighbourhood of the identity there is some fixed $K$ such that  $\Lambda_{r,s}(\cdot)[K,M]$ fails to be convex (concave).

As a consequence of \eqref{eq:Lambda-Gamma} and Theorem \ref{th:stab - gamma}, we obtain the convexity (concavity) range of the function $\Lambda$ itself. The latter is explained in Section \ref{sec3}, where we show that the trace function $\Lambda_{r,s}$ is  convex only when $r=1$ and $s\geq 1/2,$ and that it is never concave.

Independently of this work, Zhang \cite{Z21} considered the same function $\Lambda_{r,s}(A)[K,M]$ (under the notation $\Psi_{p,s}(A)$) and showed that:
	
	(1) for any $r\neq 0$ and $s>1,$ it is not concave in $A$ for all $K$ and $M$; 
	
	(2) for any $1/2\leq r<1$ and $1/2p\leq s<1,$ the function $\Lambda_{r,s}$ is not convex in $A.$ 

\noindent
We complement these results and give the complete parameter range for which $\Lambda_{r,s}$ is convex (concave).

\section{Stability of convexity of the trace functions $\Gamma_{p,s}$ and $\Psi_{p,q,s}$}
\label{sec2}
Let $\cM_n$ denote the set of complex $n\times n$ matrices, $\cP_n\subset \cM_n$ denote the subset of positive definite matrices, and $\cI_n\subset\cM_n$ denote the subset of invertible matrices.

For $r,s\in\RR,$ matrices $A, M\in \cP_n$ and $K\in \cI_n,$ the function $\Lambda$ is defined as follows
\begin{equation}\label{def:Lambda}
	\Lambda_{r,s}(A)[K, M]: = \trace\left[\{K^*A^r M A^r K\}^s\right].
\end{equation}

\begin{remark}
Note that,  due to the symmetry
\[
 \Lambda_{r,s}(A)[K,M]=\Lambda_{-r,-s}(A)[(K^*)^{-1}, M^{-1}],
\]
it is sufficient to consider the case $s>0$ when studying convexity (concavity) properties of $\Lambda_{r, s}.$
\end{remark}

\subsection{Stability of convexity of the trace function $\Gamma_{p,s}$}

As we noted in \eqref{eq:Lambda-Gamma}, fixing $M=I$ in the function $\Lambda$ results in the function $\Gamma;$ that is,
\[
 \Lambda_{r,s}(A)[K,I]=\Gamma_{2r,s}(A)[K]\ . 
\]
Therefore, when $M=I,$ Theorem~\ref{thm:gamma} provides large parameters range where $\Lambda$ is convex and concave. Our main result, Theorem \ref{th:stab - gamma}, asserts that convexity and concavity are not stable as functions of $M.$  In other words,  there are arbitrarily small perturbations of $M$ away from the identity that result in loss of concavity (or convexity) of $\Lambda.$ In the sequel we let $\|T\|$ denote the operator norm of a matrix
 $T.$

\begin{theorem}\label{th:stab - gamma} 
 For each  $r\neq 0,1$ and $s>0$ and each  $\epsilon > 0,$  there exists an invertible matrix $K\in \cI_2$ and a positive definite matrix $M\in \cP_2$  satisfying $\|M-I\|<\epsilon$  such that  $\Lambda_{r,s}(\cdot) [K,M]$  is neither  concave nor convex. In fact, $K$ and $M$ can be chosen as
\[
 K=\begin{pmatrix} 1 & 0\\0 & k \end{pmatrix}, \ \ \
  M=\begin{pmatrix} 1 &t \\t & 1 \end{pmatrix}
\]
 for suitably small positive $k$ and any choice of $|t|<\epsilon.$
\end{theorem}

\begin{corollary}
For any $\epsilon>0$ there exists a positive definite matrix $M\in \cP_2$
 satisfying $\|M-I\|<\epsilon$ such that
\begin{enumerate}[(a)]
 \item if $0< r\leq 1/2$ and $0<s\leq \frac{1}{2r},$ then  there is an invertible $K\in \cI_2,$ for which   $\Lambda_{r,s}(\cdot)[K,M]$ is 
    not concave; 
  \item  if $ r< 0$ and $s>0,$ then there is an invertible $K\in \cI_2$ such that  $\Lambda_{r,s}(\cdot)[K,M]$ is  not convex; 
 \item if  $1/2\leq r< 1$ and $s\geq \frac{1}{2r},$ then  there is an invertible $K\in \cI_2,$ for which $\Lambda_{r,s}(\cdot)[K,M]$ {is not convex}.
 \end{enumerate}
\end{corollary}

In the upcoming proof we make use of the convergent power series expansion
(generalized  binomial theorem)
\[
(I+X)^r =\sum_{k=0}^\infty \binom{r}{k} X^k,
\]
valid for $\|X\|<1$ and $r \in \RR$.

\begin{proof}[Proof of Theorem \ref{th:stab - gamma}]
 Fix $r,s,t\in \RR$ with $0\ne r\ne 1$ and $s>0$ and  $|t|<\epsilon.$ Define, for  $0\le b<\frac12,$
\[
 g_r(b):=r - \frac{2}{b-1}\Big(\Big(1+\frac{b-1}{2}\Big)^r - 1 \Big). 
\]
 In particular, $g_r$ is continuous and
\[
 g_r(0)=  r +2^{1-r} -2
\]
Computing the second derivative of $g_r(0)$ as a function of $r$ shows that this function is strictly convex, hence can have at most two zeros. They occur at $r=0$ and $r=1.$ Moreover, $g_r(0)$ is negative for $r \in (0,1)$ and positive for $r<0$ and $r>1.$ Thus, for our fixed $r\neq 0, 1,$  there exists a $b>0$ sufficiently close to $0$ such that
 $g_r(b)\ne 0.$ Fix such a $b.$

Let
\begin{equation}\label{eq:M-K}
 M=M(t)=\begin{pmatrix} 1&t\\t&1\end{pmatrix} \quad \text{ and } \quad K=\begin{pmatrix}1 &0\\0&0\end{pmatrix}. 
\end{equation}
In particular, $\|M-I\|<\epsilon.$ We discuss the adaptation to an invertible $K$ at the end of the proof.

For $0<|x|<\frac12,$ let
\[
A_1=A_1(b)=\begin{pmatrix} 1 & 0 \\0& b\end{pmatrix}, \ \ \
A_2 =A_2(x) = I+ x \begin{pmatrix}0&1\\1&0\end{pmatrix},
\]
 and 
\[
  A:= \frac{A_1+A_2}{2} = I + \frac12 \begin{pmatrix} 0&x\\x & b-1 \end{pmatrix}
  =: I+\frac12 T.
\]
 
Consider the function  $\Xi_{r,s}$ defined by 
\[
\Xi(x)= \Xi_{r,s}(x)[t,b]=\Lambda_{r,s}(A_1)[K,M] +\Lambda_{r,s}(A_2)[K,M]-2\Lambda_{r,s}(A)[K,M]\   
\]
for $|x|<\frac{1}{2}.$ Note that $\Xi(0)=0.$ We will show that there exist
 $x_\pm$ arbitrary close to $0$ such that $\pm \Xi(x_\pm)>0,$ thus 
proving $\Lambda_{r,s}$ is neither convex nor concave for all $M$ in a neighbourhood of the identity.

To prove the claim above, 
   let $c=b-1$ and  observe that, for $k\ge 2,$
\[
 T^k =\begin{pmatrix} 0&0\\0& c^k \end{pmatrix} 
  +   c^{k-1}  x \begin{pmatrix} 0&1\\1&0\end{pmatrix} 
 + x^2 T_k(x) = L_k(x) + x^2 T_k(x).
\]
for some (matrix) $T_k(x)$ that is polynomial in $x.$ 
Further, using $|x|<\frac12$ (and $|c|<1$), 
\[
 2^{-k} \|T_k(x)\| \le 2^{-k}\sum_{m=2}^k 
  \binom{k}{m} |x|^{m}\,|c|^{k-m}
 \le  2^{-k} (|x|+|c|)^k < \left(\frac34\right)^k.
\]
 Thus, letting
\[
F(x) = \sum_{k=2}^\infty \binom{r}{k} 2^{-k} T_k(x)
\]
it follows that 
\[
 \|F(x)\| \le  \sum_{k=2}^\infty\binom{r}{k} \left(\frac{3}{4}\right)^k
 \le \left(1+\frac34\right)^r
\]
 for $|x|<\frac12.$  Moreover,
\[
\begin{split}
A^r&= \Big(I+\frac12 \,T\Big)^r = \sum_{k=0}^\infty \binom{r}{k} 2^{-k} L_k(x) +x^2 F(x)\\
&  =
 I + \sum_{k=1}^\infty \binom{r}{k} 2^{-k}c^k \begin{pmatrix} 0&0\\0&1 
 \end{pmatrix} +  \frac{x}{c}  \sum_{k=1}^\infty \binom{r}{k} 
    2^{-k}c^k \begin{pmatrix}0&1\\1&0\end{pmatrix} + x^2F(x,c)\\
  &  =\begin{pmatrix} 1 & \frac{x}{c} [(1+ \frac{c}{2})^r -1] \\ 
     \frac{x}{c} [(1+ \frac{c}{2})^r - 1] & (1+\frac{c}{2})^r  \end{pmatrix} + x^2F(x) \\
& = G(x)+x^2F(x),
\end{split}
\]
 with $G(x)$ also uniformly bounded for $|x|<\frac12$ because
 $b<\frac12$ implies $\frac{1}{|c|}<2.$
 Hence, letting 
  $\alpha = \frac{1}{c}  \Big[\Big(1+\frac{c}{2}\Big)^r - 1\Big],$
 we have
\[
 \trace((K^* A^rMA^rK)^s)
  = \left (1 +  2xt\alpha + x^2[\alpha^2 + f(x)] \right)^s,
\]  
 where $f(x)$ is uniformly bounded. 
Thus, 
\[
\trace((K^* A^rMA^rK)^s)  =  1 +  2s \frac{xt}{b-1}  \Big[\Big(1+\frac{b-1}{2}\Big)^r - 1\Big] + O(x),
\]
 where $\lim_{x\to 0} \frac{O(x)}{x} =0.$

A similar argument with
\[
 A_2 = I + x \begin{pmatrix}0&1\\1&0\end{pmatrix}
\]
gives
\[
A_2^r \sim  \begin{pmatrix} 1 & rx \\ rx & 1\end{pmatrix}
\]
and thus
\[
 \trace((K^* A_2^r M A_2^r K)^s) = 1 +2strx + O_2(x),
\]
 where $\lim_{x\to 0} \frac{O_2(x)}{x}=0.$
Hence,
\[
\begin{split}
\Xi(x)=& \trace((K^* A_1^rMA_1^rK)^s) +  \trace((K^* A_2^r M A_2^r K)^s)
 -2 \trace((K^* A^rMA^rK)^s  \\ 
  = &  1 +  (1+ 2strx)  -2
   \left ( 1 +  2s \frac{xt}{b-1}  \Big[\Big(1+\frac{b-1}{2}\Big)^r - 1\Big]\right) + O(x)\\
  & =   2 sxt  \Big[ r - \frac{2}{b-1} \Big(\Big(1+\frac{b-1}{2}\Big)^r - 1 \Big) \Big]\ +O(x) \\
 & = 2st x g_r(b) +O(x),
\end{split}
\]
where $O(x)$ is a function satisfying $\lim_{x\to 0} \frac{O(x)}{x}=0.$
 Hence choosing $x$ near $0$ with opposite signs concludes the proof
 with the present choice of (not so invertible) $K.$

It remains to show that $K$ can be taken to be invertible. Let $A,A_1,A_2$ and $M$ be as defined above and for $k > 0$ set 
\[
K = K(k) = \begin{pmatrix} 
1 & 0 \\
0 & k
\end{pmatrix}.
\]
Define 
\[
\wXi(x,k)=\Lambda_{r,s}(A_1)[K(k),M] +\Lambda_{r,s}(A_2)[K(k),M]-2\Lambda_{r,s}(A)[K(k),M]\  . 
\]
Since $s>0,$ one has that $\Xi(x,k)$ is continuous in $k$ and that
\[
\lim_{k \to 0} \wXi(x,k) = \wXi(x,0) = \Xi(x)\ .
\]
Thus, for sufficiently small choices of $k,$ the sign of $\wXi(x,k)$ may be determined by choosing $x$ as above and then  $k>0$ sufficiently small.
 \end{proof}

\subsection{Stability of convexity of the trace function $\Psi_{p,q,s}$} In this subsection we show that the stability of convexity (and concavity) of $\Gamma_{p,s}$ and $\Psi_{p,q,s}$ are closely related.
As noted in \eqref{eq:Lambda-Psi}, in a certain parameter range, the function $\Psi$ can be obtained  from $\Lambda$ by taking $K=I;$ that is,
\[
 \Lambda_{r,s}(A)[I, M]=\Psi_{1,2r,s}(M,A)[I]. 
\]
Hence, when $K=I,$ Theorem~\ref{thm:Psi-conv} identifies  large parameter ranges where $\Lambda$ is convex and concave. We show as a corollary of Theorem \ref{th:stab - gamma} that this condition on $K$ is not stable in the sense that perturbing $M$ away from the identity even a {\it little}  results in the loss of convexity (or concavity) of $\Lambda.$

\begin{corollary}\label{cor:stab - Psi}
  For any $\epsilon  > 0$ there exists an invertible matrix $K\in \cI_2$
and a positive definite matrix $M\in \cP_2$ 
 such that $\|K-I\|<\epsilon$ and  $\Lambda_{r,s}(A)[K,M]$ is not 
\begin{enumerate}[(a)]
\item concave for $0\leq r\leq  1/2$ and $0<s\leq \frac{1}{1+2r},$
\item convex for $-1/2< r\leq 0,$ and $s\geq \frac{1}{1+2r}.$
\end{enumerate}
\end{corollary}

\begin{proof}
The definition of $\Lambda_{r,s}$ and the fact that the eigenvalues of the matrices $XY$ and $YX$ are the same imply 
\begin{equation}\label{eq: lambda-sim}
	\Lambda_{r,s}(A)[K, M^*M]=\Lambda_{r,s}(A)[M, KK^*].
\end{equation}

Now let 
\begin{equation*}\label{eq:K}
\widetilde{K} = \frac{1}{2}\begin{pmatrix} \sqrt{1+t}+\sqrt{1-t} & \sqrt{1+t}-\sqrt{1-t} \\ \sqrt{1+t}-\sqrt{1-t} & \sqrt{1+t}+\sqrt{1-t} \end{pmatrix}\quad \text{ and } \quad  
\widetilde{M} = \begin{pmatrix} 1 & 0 \\ 0 & k^2 \end{pmatrix}.
\end{equation*}
Note that one can choose $0<t<1$ such that $\|\widetilde{K}-I\| < \epsilon.$
Let $K$ and $M$ be as in the proof of Theorem \ref{th:stab - gamma} and observe that
\[
\widetilde{M} = K^{\ast}K \quad \text{ and } \quad \widetilde{K}\widetilde{K}^{\ast} = \begin{pmatrix} 1 & t \\ t & 1 \end{pmatrix} = M.
\]
Then by \eqref{eq: lambda-sim},
\[
 \Lambda_{r,s}(A)[\widetilde{K}, \widetilde{M}]=\Lambda_{r,s}(A)[\widetilde{K}, K^*K]=\Lambda_{r,s}(A)[K, \widetilde{K}\widetilde{K}^*]=\Lambda_{r,s}(A)[K, M] \ .
\]
Invoking Theorem \ref{th:stab - gamma} completes the proof.
\end{proof}

\section{Convexity and concavity of the trace function $\Lambda_{r,s}$}\label{sec3}
In this section we collect the results concerning the convexity and concavity ranges of $\Lambda_{r,s}.$ The main result here is Theorem \ref{thm:main}.
 It  states that  $\Lambda_{r,s}(A)[K,M]$ is convex only for $r=1$ and $s\geq 1/2$, and that it is never concave. We will see that the questions of convexity and concavity for $\Lambda_{r,s}(\cdot)[K,M]$ for most parameter ranges were already settled in Theorem~\ref{th:stab - gamma} and
 Theorem~\ref{thm:gamma}.

\begin{theorem}\label{thm:main}
   Fix real numbers $r,s$ with $r\ne 1$ and $s>0$  and an integer  $n\geq 2.$ The function
	\[
	\cP_n \ni A\mapsto \Lambda_{r,s}(A)[K,M] 
	\]
	is convex for all $K\in\cI_n$ and $M\in\cP_n$ if and only if $r=1$ and $s\geq 1/2;$
         and is never concave for all $K\in\cI_n$ and $M\in\cP_n.$
\end{theorem}

The proof of Theorem~\ref{thm:main} is divided into three parts. Proposition~\ref{prop: lambda-cvx} below asserts the convexity of $\Lambda_{r,s}$ for $r=1$ and $s\geq 1/2.$ Theorem~\ref{th:stab - gamma} disproves $\Lambda_{r,s}$ is neither  convex nor concavity  for all $r\neq 0, 1$ and $s>0.$ Finally, Proposition~\ref{prop:not-convex-r=1}, based upon Theorem~\ref{thm:gamma}, shows $\Lambda_{r,s}$ is neither convex nor concave when $r=1$ and $0<s<1/2.$

\begin{proposition}\label{prop: lambda-cvx}
	Fix a positive integer $n,$ $K\in\cM_n$ and $M \in \cP_n.$ If $s\ge \frac12,$ then the function
	\[
	\cP_n \ni A\mapsto \Lambda_{1,s}(A)[K,M]
	\]
	is convex in $A.$
\end{proposition}
\begin{proof}
	
	The function $\Lambda_{1,s}$ can be expressed in terms of the function $\Gamma,$ for which the convexity range is known. Explicitly,
\begin{equation}\label{eq:lambda-gamma}
	\Lambda_{1,s}(A)[K,M] =\Lambda_{1,s}(M^{1/2}AM^{1/2})[M^{-1/2}K,I]=\Gamma_{2,s}(M^{1/2}AM^{1/2})[M^{-1/2}K]\ .
	\end{equation}
	From Theorem \ref{thm:gamma}, for  $s \in [1/2,\infty),$ the mapping $A\mapsto \Gamma_{2,s}(M^{\frac12} A M^{\frac12})[M^{-\frac12}K]$
         is convex. 
\end{proof}

\begin{proposition}\label{prop:not-convex-r=1}
 Given an integer $n\ge2,$ there exists an invertible matrix $K\in\cI_n$ such 
 that  if $0<s<\frac12,$ 
   then there exist  positive definite matrices $M_{\pm} \in\cP_n$ so that the function
\[
\cP_n \ni A\mapsto \Lambda_{1,s}(A)[K,M]
\]
is not convex in $A$ with $M=M_+$ and  not concave in $A$ with $M=M_-.$
\end{proposition}

\begin{proof}
It suffices to prove the result with $n=2.$ In this case choose
\[
   K=\begin{pmatrix} 1&0\\0&0\end{pmatrix}, \ \ \ M=M(t)=\begin{pmatrix} 1&t\\t&1\end{pmatrix}, 
\]
 for $|t|<1$ to be chosen later.

Following the proof of Theorem \ref{th:stab - gamma} consider
\[
A_1=\begin{pmatrix} 1 & 0 \\0& 1\end{pmatrix}, \ 
A_2 =A_2(x) = I+ x \begin{pmatrix}0&1\\1&0\end{pmatrix}\ 
\]
for $|x|<\frac{1}{2}.$  Let $A = \frac{A_1+A_2}{2}.$ Expansion up to second order terms in $x$ gives
\begin{align*}
	\trace((K^* A_1 M A_1 K)^s) &= 1,\\
	\trace((K^* A_2 M A_2 K)^s) &= \big(1 + 2tx + x^2\big)^s \sim 1 + 2stx + sx^2 + 2s(s-1)t^2x^2,\\
    \trace((K^* A M A K)^s)  &  = \bigg(1 + tx + \frac{x^2}{4}\bigg)^s \sim 1 + stx + s\frac{x^2}{4} + \frac{s(s-1)}{2}t^2x^2. 
\end{align*}
Hence
\begin{align*}
\Xi_{1,s}(t,x)=& \trace((K^* A_1 MA_1 K)^s) +  \trace((K^* A_2 M A_2 K)^s)
-2 \trace((K^* A MA K)^s)  \\ 
&\sim  \frac12  x^2\bigg(s^2t^2 + s\Big(\frac{1}{2}-t^2\Big)\bigg) = x^2\, h(s, t).
\end{align*}
Thus for $x$ sufficiently small, the sign of $\Xi_{1,s}(t,x)$ is determined by the sign of $h(s, t).$ Since $s>0,$ the latter is positive for all $t$ such that $t^2 < \frac{1}{2}.$ Thus, for such $t$ and every $s>0,$ the function  $A\mapsto\Lambda_{1,s}(A)[K,M(t)]$  is not concave.
Moreover, the roots of $h(s, t)$ viewed as a quadratic function in $s$ are $0$ and $1 - \frac{1}{2t^2}.$ Therefore $h(s, t)$ is negative for $1/2<t^2<1$ and $0 < s < 1 - \frac{1}{2t^2}<1/2.$ Hence, for every fixed $0 < s < \frac{1}{2}$ there is a $|t|<1$ such that $h(s, t)<0.$ Therefore,  
 for each $0<s<\frac12$ there exists a $|t|<1$ such that the map $A\mapsto \Lambda_{1,s}(A)[K,M(t)]$ is not convex. 
 
 The matrix $K$  can be chosen invertible by the same argument as presented at the end of Theorem \ref{th:stab - gamma}.
\end{proof}

\begin{remark}\rm
	(a) Note that even though $\Lambda_{1,s}(A)[I,M]$ is convex in $A$ for fixed positive definite $M$ and  $s\geq \frac12$ by Proposition~\ref{prop: lambda-cvx}, it is not jointly convex in $A$ and $M$ for fixed $s\geq\frac12$  as can be seen by combining the duality between the functions $\Lambda$ and $\Psi$ in \eqref{eq:Lambda-Psi} and  Theorem~\ref{thm:Psi-conv-nec} item~\eqref{i:Psi-conv-nec2}.
  That $\Lambda_{1,s}(A)[I,M]$ is not jointly convex for fixed $s\ge 0$ is also a consequence of the following example that has the  advantage of involving $2\times 2,$  rather than  $4\times 4,$ matrices as was done in Theorem \ref{thm:gamma-necessary}. Let 
	\[
	A_1 = \begin{pmatrix}
		1/2 & 0 \\ 
		0 & 1
	\end{pmatrix},
	\qquad
	A_2 = \begin{pmatrix}
		1 & 0 \\ 
		0 & 1/2
	\end{pmatrix},
	\qquad
	M_1 = \begin{pmatrix}
		4 & 0 \\ 
		0 & 1
	\end{pmatrix},
	\qquad
	M_2 = \begin{pmatrix}
		1 & 0 \\ 
		0 & 4
	\end{pmatrix}.
	\]
	Direct computation shows 
	\[
	\Lambda_{1,s}(A_1)[I_2,M_1]/2+\Lambda_{1,s}(A_2)[I_2,M_2]/2-\Lambda_{1,s}\left(\frac{A_1+A_2}{2}\right)\left[I_2,\frac{M_1+M_2}{2}\right] = 2 - 2\left(\frac{45}{32}\right)^s,
	\]
	which is negative for $s>0.$
	
	(b) We mentioned that Zhang \cite{Z21} shows joint convexity of the function 
	$$(A,B,C)\mapsto\trace\left| B^{-p}K_1AK_2C^{-q}\right|^{s}$$ for $0< p,q\leq 1/2$ such that $p+q<1$ and $s\geq 1/(1-p-q).$ Moreover, he shows that these conditions are optimal in the sense that if the function $(A,B,C)\mapsto\trace\left| B^{-p}AC^{-q}\right|^{s}$ is jointly convex in any finite dimension, then $p,q,s$ must satisfy the above conditions.
\end{remark}

\vspace{0.3in}

\textbf{Acknowledgements.} All authors are thankful to the American Institute of Mathematics for sponsoring the workshop ``Noncommutative inequalities" in June 2021 as well as the organizers of the workshop, where the collaboration between the authors and the work described in this paper began.


\begin{thebibliography}{10}

\bibitem{A79} Ando, T. (1979). Concavity of certain maps on positive definite matrices and applications to Hadamard products. Linear algebra and its applications, 26, 203-241.

\bibitem{AD15} Audenaert, K. M., \& Datta, N. (2015). $\alpha-z$-R\'enyi relative entropies. Journal of Mathematical Physics, 56(2), 022202.

\bibitem{B13} Beigi, S. (2013). Sandwiched R\'enyi divergence satisfies data processing inequality. Journal of Mathematical Physics, 54(12), 122202.

\bibitem{B04} Bekjan, T. N. (2004). On joint convexity of trace functions. Linear algebra and its applications, 390, 321-327.

\bibitem{B97}  Bhatia, Matrix analysis, Springer-Verlag, New York, 1997

\bibitem{BDK08} Bjelakovic, I., Deuschel, J. D., Kr\"uger, T., Seiler, R., Siegmund-Schultze, R., \& Szkoa, A. (2008). Typical support and Sanov large deviations of correlated states. Communications in mathematical physics, 279(2), 559-584.

\bibitem{CFL16} Carlen, E. A., Frank, R. L., \& Lieb, E. H. (2016). Some operator and trace function convexity theorems. Linear Algebra and its Applications, 490, 174-185.

\bibitem{CFL18} Carlen, E. A., Frank, R. L., \& Lieb, E. H. (2018). Inequalities for quantum divergences and the Audenaert-Datta conjecture. Journal of Physics A: Mathematical and Theoretical, 51(48), 483001.

\bibitem{CL08} Carlen, E. A., \& Lieb, E. H. (2008). A Minkowski type trace inequality and strong subadditivity of quantum entropy II: convexity and concavity. Letters in Mathematical Physics, 83(2), 107-126.

\bibitem{Che20} Chehade, S. (2020). Saturating the Data Processing Inequality for $\alpha-z $ R\'enyi Relative Entropy. arXiv preprint arXiv:2006.07726.

\bibitem{E73} Epstein, H. (1973). Remarks on two theorems of E. Lieb. Communications in Mathematical Physics, 31(4), 317-325.

\bibitem{FL13} Frank, R. L., \& Lieb, E. H. (2013). Monotonicity of a relative R\'enyi entropy. Journal of Mathematical Physics, 54(12), 122201.

\bibitem{H13} Hiai, F. (2013). Concavity of certain matrix trace and norm functions. Linear algebra and its applications, 439(5), 1568-1589.

\bibitem{H16} Hiai, F. (2016). Concavity of certain matrix trace and norm functions. II. Linear Algebra and its Applications, 496, 193-220.

\bibitem{HMH09}  Hiai, F., Mosonyi, M., \& Hayashi, M. (2009). Quantum hypothesis testing with group symmetry. Journal of mathematical physics, 50(10), 103304.

\bibitem{HMPB11} Hiai, F., Mosonyi, M., Petz, D., \& B\'eny, C. (2011). Quantum f-divergences and error correction. Reviews in Mathematical Physics, 23(07), 691-747.

\bibitem{HP91} Hiai, F., \& Petz, D. (1991). The proper formula for relative entropy and its asymptotics in quantum probability. Communications in mathematical physics, 143(1), 99-114.

\bibitem{LRD17} Leditzky, F., Rouz\'e, C., \& Datta, N. (2017). Data processing for the sandwiched R\'enyi divergence: a condition for equality. Letters in Mathematical Physics, 107(1), 61-80.

\bibitem{L73} Lieb, E. H. (1973). Convex trace functions and the Wigner-Yanase-Dyson conjecture. Les rencontres physiciens-math\'ematiciens de Strasbourg-RCP25, 19, 0-35.

\bibitem{LR73} Lieb, E. H., \& Ruskai, M. B. (1973). Proof of the strong subadditivity of quantum-mechanical entropy. Les rencontres physiciens-math\'ematiciens de Strasbourg-RCP25, 19, 36-55.

\bibitem{LR74} Lieb, E. H., \& Ruskai, M. B. (1974). Some operator inequalities of the Schwarz type. Advances in Mathematics, 12(2), 269-273.


\bibitem{L75} Lindblad, G. (1975). Completely positive maps and entropy inequalities. Communications in Mathematical Physics, 40(2), 147-151.

\bibitem{MDSFT13} M\"uller-Lennert, M., Dupuis, F., Szehr, O., Fehr, S., \& Tomamichel, M. (2013). On quantum R\'enyi entropies: A new generalization and some properties. Journal of Mathematical Physics, 54(12), 122203.

\bibitem{P86-1} Petz, D. (1986). Quasi-entropies for finite quantum systems. Reports on mathematical physics, 23(1), 57-65.

\bibitem{P86} Petz, D. (1986) Sufficient subalgebras and the relative entropy of states of a von Neumann algebra. Communications in Mathematical Physics, 105:123?131.

\bibitem{P88} Petz, D. (1988). Sufficiency of channels over von Neumann algebras. The Quarterly Journal of Mathematics, 39(1), 97-108.

\bibitem{U54} Umegaki, H. (1954). Conditional expectation in an operator algebra. Tohoku Mathematical Journal, Second Series, 6(2-3), 177-181.

\bibitem{WWY14} Wilde, M. M., Winter, A.,  Yang, D. (2014). Strong converse for the classical capacity of entanglement-breaking and Hadamard channels via a sandwiched R\'enyi relative entropy. Communications in Mathematical Physics, 331(2), 593-622.

\bibitem{Z20} Zhang, H. (2020). From Wigner-Yanase-Dyson conjecture to Carlen-Frank-Lieb conjecture. Advances in Mathematics, 365, 107053.

\bibitem{Z20-2} Zhang, H. (2020). Equality conditions of data processing inequality for $\alpha-z$ R\'enyi relative entropies. Journal of Mathematical Physics, 61(10), 102201.

\bibitem{Z21} Zhang, H. (2021). Some convexity and monotonicity results of trace functionals. arXiv preprint arXiv:2108.05785.

\end{thebibliography}
\end{document}